\documentclass[a4paper,UKenglish,cleveref, autoref, thm-restate]{lipics-v2021}
\hideLIPIcs  %uncomment to remove references to LIPIcs series (logo, DOI, ...), e.g. when preparing a pre-final version to be uploaded to arXiv or another public repository

\usepackage{comment} % enables the use of multi-line comments (\ifx \fi)
\usepackage{amsmath,amssymb,amsthm}  % assumes amsmath package installed
\usepackage{graphicx}
\usepackage{caption, subcaption}
\usepackage{verbatim}
\usepackage{float}
\usepackage{xcolor}
\usepackage{mdframed}
\usepackage{indentfirst}
\usepackage[T1]{fontenc}
\usepackage[utf8]{inputenc}
\renewcommand{\qed}{\quad\qedsymbol}
\newcommand{\R}{\mathbb{R}}

\newcommand{\T}{\mathcal{T}}
\newcommand{\MST}{\mathrm{MST}}
\title{Euclidean Steiner Shallow-Light Trees in Higher Dimensions} %TODO Please add

\author{Devin Frost}{California State University Northridge, Los Angeles, CA, USA}{devin.frost.481@my.csun.edu}{}{}%TODO mandatory, please use full name; only 1 author per \author macro; first two parameters are mandatory, other parameters can be empty. Please provide at least the name of the affiliation and the country. The full address is optional. Use additional curly braces to indicate the correct name splitting when the last name consists of multiple name parts.

\author{Kimberly Kokado}{California State University Northridge, Los Angeles, CA, USA}{kimberly.kokado.43@my.csun.edu}{}{}

\author{Csaba D. T\'oth}{California State University Northridge, Los Angeles, CA, USA}{csaba.toth@csun.edu}{}{}

\authorrunning{D.~Frost, K.~Kokado, and C.~D.~T\'oth} %TODO mandatory. First: Use abbreviated first/middle names. Second (only in severe cases): Use first author plus 'et al.'

\Copyright{Devin Frost and Kimberly Kokado} %TODO mandatory, please use full first names. LIPIcs license is "CC-BY";  http://creativecommons.org/licenses/by/3.0/

 %TODO mandatory; please add comma-separated list of keywords

\category{} %optional, e.g. invited paper

\relatedversion{} %optional, e.g. full version hosted on arXiv, HAL, or other respository/website
\nolinenumbers %uncomment to disable line numbering

\EventEditors{John Q. Open and Joan R. Access}
\EventNoEds{2}
\EventLongTitle{42nd Conference on Very Important Topics (CVIT 2016)}
\EventShortTitle{CVIT 2016}
\EventAcronym{CVIT}
\EventYear{2016}
\EventDate{December 24--27, 2016}
\EventLocation{Little Whinging, United Kingdom}
\EventLogo{}
\SeriesVolume{42}
\ArticleNo{23}
\begin{document}

\maketitle

%TODO mandatory: add short abstract of the document
\begin{abstract}
 This paper proves a conjecture by Solomon about Steiner shallow-light trees (SLT) in Euclidean $d$-space: It is shown that for any finite point set $\mathbb{R}^d$, any root, and any $\epsilon>0$, there is a Euclidean Steiner $(1+\epsilon,O(\sqrt{1/\epsilon}))$-SLT without any dependence on dimension. We also revisit the core example, designed by Solomon, in the plane and its generalization to $d$-space.
\end{abstract}

\section{Introduction}
\label{sec:itro}

For an edge-weighted graph $G=(V,E,w)$, shallow-light trees (SLT) were introduced in the 1990s~\cite{AwerbuchBP90,CongKRSW91,KRY95} to  simultaneously approximate  shortest path trees (SPT) and minimum spanning trees (MST). For a root $s\in V$, an $(\alpha,\beta)$-SLT
is a spanning tree $T$ of $G$ with \emph{root stretch} at most $\alpha$ (that is, for every vertex $v\in V$, we have $d_T(s,v)\leq \alpha \cdot d_G(s,v)$) and \emph{lightness} at most $\beta$ (that is, $w(T)\leq \beta\cdot w(\MST(G))$. It is known that for every $\epsilon>0$, one can construct a $(1+\epsilon,O(1/\epsilon))$-SLT in linear time~\cite{AwerbuchBP90,CongKRSW91,KRY95}. The trade-off between root stretch and lightness is the best possible up to constant factors, and is already tight in planar graphs~\cite{KRY95}. Elkin and Solomon~\cite{ElkinS15} showed that the trade-off is the best possible over Euclidean metrics in the plane, as well, where $G$ is a complete graph on $n$ points in $\mathbb{R}^2$, and the edge weights are the Euclidean distances between the end points. Solomon~\cite{Solomon15} showed that the use of Steiner points can dramatically improve this bound and constructed, for any finite point set $V\subset \mathbb{R}^2$, source $s\in V$ and $\epsilon>0$, a Steiner $(1+\epsilon,O(\sqrt{1/\epsilon}))$-SLT. This trade-off is also the best possible: It is tight for a set $V$ of $\Theta(\sqrt{1/\epsilon})$ evenly spaced points on a circle, and any root $s\in V$. Very recently, Hung et al.~\cite{LST26} presented a bicriteria approximation algorithm for finding instance-optimal SLTs and Steiner SLTs for a given set of $n$ points and root in Euclidean plane: by relaxing the root stretch to $1+O(\epsilon \log \frac{1}{\epsilon})$, they can approximate the minimum weight of $(1+\epsilon,O(1/\epsilon))$-SLT within a factor of $O(\log^2\frac{1}{\epsilon})$.

Steiner SLTs with root stretch $1+\epsilon$ were the key building block to construct
Euclidean Steiner $(1+\epsilon)$-spanners for $n$ points in the plane with optimal lightness $O(1/\epsilon)$~\cite{bhore2022euclidean}. However, in dimensions $d\geq 3$, there is still a gap between the lower bound $\Omega(1/\epsilon^{d/2})$ and the upper bound $\tilde{O}(1/\epsilon^{(d+1)/2})$ for the minimum lightness of Euclidean Steiner $(1+\epsilon)$-spanners~\cite{bhore2022euclidean,LeS23}. Solomon~\cite{Solomon14} sketched a higher dimensional generalization of his SLT construction, which would yield an $(1+\epsilon,O(1/\epsilon)\cdot 2^d)$-SLT in Euclidean $d$-space, however he conjectured
that the dependence on the dimension $d$ is not needed.

\smallskip\noindent\textbf{Results.} The main result of this paper (\Cref{thm:folding} in \Cref{sec:folding}) proves Solomon's conjecture. We construct, for every finite set in $\mathbb{R}^d$ in any dimension $d\geq 2$, a Euclidean Steiner $(1+\epsilon,O(\sqrt{1/\epsilon}))$-SLT without any dependence on $d$: We unfold a 2-dimensional surface generated by the Euclidean MST into the plane, construct a Steiner SLT in the plane, and then lift the tree back into $\mathbb{R}^d$ to obtain a Steiner SLT for the given point set.
In Section~\ref{sec:Core}, we  revisit Solomon's construction in the plane, provide additional details (that were missing in~\cite{Solomon14,Solomon15}), and adjust key parameters of the construction (as the lightness and root stretch analysis do not go through with the parameters originally proposed in~\cite{Solomon14,Solomon15}). In Section~\ref{sec:pyramid}, we consider the straightforward generalization of Solomon's core example, as outlined in~\cite[Section~3]{Solomon14}, and confirm the dependence on the dimension $d$ in this case.

\section{Folding Analysis in $d$-Dimensions}
\label{sec:folding}

In this section, we construct a Steiner SLT for point set in $\mathbb{R}^d$, for any dimension $d$, with the same bounds as for $d=2$ in Solomon~\cite{Solomon15}. Our main result is the following theorem.

\begin{theorem}\label{thm:folding}
    Let $V$ be a finite set in $\mathbb{R}^d$, and $s\in V$ a distinguished source point. Then we can construct a Steiner $(1+\epsilon,O(\sqrt{1/\epsilon}))$-SLT for $V$ rooted at $s$.
\end{theorem}

\subparagraph{The Surface.}
Let $M$ be an MST of $V$, rooted at $s\in V$.
Let $H=(s=v_1,\ldots , v_n)$ be the Hamiltonian path that visits the vertices in
in the order in which they are discovered by the DFS traversal of $M$ starting from $s$. It is well known that $w(H)\leq 2\cdot w(M)$. We consider $H$ a geometric path, where each edge is a line segment between its endpoints.
We select a set of (Steiner) points $B=\{b_1,\ldots ,b_k\}$, called \emph{break points}, at vertices or in the interior of edges of $H$ as follows. The first break point $b_1$ is $v_1$ and the final break point $b_k$ is $v_n$. Then, if $b_i$ is already defined, then we choose the next  break point $b_{i+1}$ as a Steiner point on $H$ such that
\[
d_H(b_i,b_{i+1})=\sqrt{\epsilon}\cdot d(s,b_{i+1}),
\]
where $d_H(u,v)$ is the distance between the points $u$ and $v$ along $H$. Subdivide the edges of $H$ with the break points, if necessary, and denote the resulting (Steiner) path by $H^*=(h_1,\ldots , h_{n+k})$.

For each edge $h_j h_{j+1}$ of $H^*$, define the 2-dimensional cone $C_j$ as the convex hull of the rays $\overrightarrow{sh_j}$ and $\overrightarrow{sh_{j+1}}$.
For each subpath of $H^*$ between consecutive breakpoints, we combine the cones corresponding to the edges into a piecewise-linear surface:
For $i=1,\ldots ,k$, we construct the surface $S_i$ by successively gluing together the cones $C_j$ corresponding to the edged $h_jh_{j+1}$ between breakpoints $b_i$ and $b_{i+1}$. Recall that cones $C_j$ and $C_{j+1}$ share $\overrightarrow{sh_{j+1}}$. So, each $S_i$ is composed of a sequence of cones $(C_{j(i)},C_{j(i)+1}\ldots,C_{j(i+1)})$, where consecutive cones are glued together at their shared boundary ray.

\subparagraph{Unfolding into the plane.}
Define the family of functions $f_i:S_i\to \R^2$ as follows, and assume for notational convenience that $S_i$ is composed of the cones $C_1,\ldots , C_\ell$. We define the function $f_i$ piecewise, as an isometry on each cone. Let $f_i$ map $C_1$ isometrically to the plane. When $f_i$ is already defined on a cone $C_j$ (hence it is defined on the ray $\overrightarrow{sh_{j+1}}$), we define $f_i$ on $C_{j+1}$ so that it maps $C_{j+1}$ isometrically into the plane, $f_{i+1}(p)=f_i(p)$ for all $p\in \overrightarrow{sh_{j+1}}$, and the cones $f_i(C_j)$ and $f_i(C_{j+1})$ are on opposite sides of $f(\overrightarrow{sh_{j+1}})$.
To ensure that $f_i$ is injective and $f_i(S_i)$ is convex, we need that the total angle of the resulting cone $f_i(S_i)$ is less than $\pi$.

\begin{lemma}
\label{lem:angle}
     For every $i=1,\ldots , k$, we have $\alpha_i\leq O(\sqrt{\epsilon})$, where  $\alpha_i$ is the angle of the cone $f_i(S_i)$, or equivalently, the sum of the angles $\alpha_j$ of each cone that makes up $f_i(S_i)$.
\end{lemma}

%The proof of \Cref{lem:angle,lem:isometry}, as well as all further details, are deferred to the appendix. Once unfolded, we construct a Steiner $(1+\epsilon,O(\sqrt{1/\epsilon}))$-SLT in the plane using Solomon's planar construction, and then fold the tree back into $\mathbb{R}^d$ to obtain a Steiner $(1+\epsilon,O(\sqrt{1/\epsilon}))$-SLT for $V$, without any dependence on the dimension $d$.

\begin{proof}
Let $i\in \{1,\ldots , k\}$ be fixed.
Let the subpath of $H^*$ between breakpoints $b_i$ and $b_{i+1}$ be re-indexed as $(p_0=b_i,p_1,\ldots,p_{m-1},p_m=b_{i+1})$. Therefore, each angle $\alpha_j$ is the angle of the triangle $\Delta(s,p_{j-1},p_j)$. Let $d(p_{j-1},p_j)=a_j$, and $\beta_j$ be the angle opposite each $d(s,p_{j-1})$. Also note that by the triangle inequality, for each $j$, we have
\[
    d(s,p_{j-1})
    \geq d(s,b_{i+1})-\sum a_j
    =d(s,b_{i+1})-\sqrt{\epsilon} d(s,b_{i+1})
    =(1-\sqrt{\epsilon})d(s,b_{i+1}).
\]
 Then we have
\begin{align*}
    \alpha &= \sum \alpha_j \\
    &\leq \frac{1}{2}\sum\sin\alpha_j  \\
    &= \frac{1}{2}\sum \sin\beta_j\frac{a_j}{d(s,p_{j-i})} &\text{Law of Sines}\\
    &\leq \frac{1}{2}\sum \frac{a_j}{d(s,p_{j-i})} &\sin\beta_i\leq 1 \\
    &\leq  \frac{1}{2}\sum \frac{a_j}{(1-\sqrt{\epsilon})d(s,b_{i+1})} &\text{Lower Bound for $d(s,p_{j-i})$} \\
    &= \frac{1}{2(1-\sqrt{\epsilon})d(s,b_{i+1})}\sum a_i \\
    &= \frac{1}{2(1-\sqrt{\epsilon})d(s,b_{i+1})}\sqrt{\epsilon} d(s,b_{i+1}) \\
    &= \frac{\sqrt{\epsilon}}{2(1-\sqrt{\epsilon})} . \qedhere
\end{align*}
\end{proof}
So, we have shown that the total angle of each $f_i(S_i)$ is less than or equal to $\frac{\sqrt{\epsilon}}{2-2\sqrt{\epsilon}}$, which is less than $\pi$ for sufficiently small $\epsilon>0$. Therefore each $f_i(S_i)$ is injective and convex.

\begin{lemma}
\label{lem:isometry}
    For every $i=1,\ldots, k$, let $d_{S_i}$ be the shortest path distance in the surface $S_i$, and let $d_2$ be the Euclidean distance in the plane. Then $f_i: S_i\to \R^2$ is an isometry between $d_{S_i}$ and $d_2$.
\end{lemma}
\begin{proof}
%Finally, we need to show that the distance between two points $x,y\in S_i$ is equal to the distance of their images in $f_i(S_i)$.
We distinguish between two cases: the distances between points that are on the same cone $C_j$ and the distances between points that are on different cones. In the case where two points are on the same cone, the distance is clearly preserved since the function is isometric on each cone. The more challenging case is that points are on different cones. Consider points $x$ and $y$ on different cones, and let $f(z_1),\ldots ,f(z_n)$ be the sequence of points where shortest curve from $f(x)$ to $f(y)$ in $f_i(S_i)$ crosses a ray on the boundary of two consecutive cones. Then we have
\begin{align*}
    d_2(f(x),f(y)) &= d_2(f(x),f(z_1))+d_2(f(z_1),f(z_2))+\ldots+d_2(f(z_n),f(y)) \\
     &= d_{S_i}(x,z_1)+d_{S_i}(z_1,z_2)+\ldots+d_{S_i}(z_n,y) &\text{Isometry within $C_j$} \\
     &\geq d_{S_i}(x,y) . &\text{Triangle Inequality}
\end{align*}

Conversely, let $z'_1,\ldots ,z'_n$ be the sequence of points where the shortest path from $x$ to $y$ in $S_i$ crosses the boundaries between cones. Then, we have
\begin{align*}
    d_{S_i}(x,y) &= d_{S_i}(x,z'_1)+d_{S_i}(z'_1,z'_2)+\ldots +d_{S_i}(z'_n,y) \\
    &= d_2(f(x),f(z'_1))+d_2(f(z'_1),f(z'_2))+\ldots +d_2(f(z'_n),f(y)) &\text{Isometry within $C_i$} \\
    &\geq d_2(f(x),f(y)) . &\text{Triangle Inequality}
\end{align*}

Therefore, we have $d_2(f(x),f(y))\geq d_{S_i}(x,y)$ and $d_{S_i}(x,y)\geq d_2(f(x),f(y))$, which can only be the case when they are equal. %For the $(n+1)$st cone, we simply add the distances $d_2(f(z_n),f(z_{n+1}))$ and $d_{S_i}(z'_n,z'_{n+1})$ respectively; and the final addend starts at $f(z_{n+1})$ and $z'_{n+1}$. Therefore,
We have shown that the functions $f_i$ are distance preserving. Together with what was proved earlier, it is clear that each $f_i$ is an unfolding of $S_i$ into $\R^2$.
\end{proof}

\subsection{Constructing the Steiner SLT}

\subparagraph{Preliminaries.}
We need a few prerequisites for the construction of the SLT. For clarity, let $H_i^*$ be the path in $H^*$ between the breakpoints $b_i$ and $b_{i+1}$, and let $V_i$ be the set of points from $V$ that are included in each $S_i$. We need a second level of breakpoints along $H^*_i$ so that we have a short path to each $v\in f_i(V_i)$, and secondly we need to meet all of the requirements for $T^{\rm core}$.

Fix $i$; the construction will be the same for each $f_i(S_i)$. We will start with a second level of breakpoints between $b_i$ and $b_{i+1}$. By our definition, the length of $H_i^*$ is $\sqrt{\epsilon}d_d(s,b_{i+1})$, so if we place $\vartheta=\lceil\sqrt{1/\epsilon}\rceil$ equidistant breakpoints $b_{i,1},\ldots ,b_{i,\vartheta}$ along the path we will have a distance in $H_i^*$ of $\epsilon d_d(s,b_{i+1})$ between each secondary break point. Denote this set by $B_i=\{b_{i,q}:q\in\{1,\dots,\vartheta\}\}$.

Let $r$ be the point in $f_i(V_i)$ closest to $s$, and let $\ell$ be a line segment through $r$ that is perpendicular to the bisector of the cone $f_i(S_i)$, and terminates at the edges $(s,b_i)$ and $(s,b_{i+1})$. Since we choose $r$ to be the closest point to $s$, all other points are in the closed halfplane determined by $\ell$ that does not include $s$. So for any point on $H^*_i$, the straight line connecting $s$ and $v$ (or $b_{i,q}$) intersects $\ell$ at some point that we denote $\tilde v$ (respectively, $\tilde b_{i,q}$). Furthermore, we will distribute $\lceil\sqrt{1/\epsilon}\rceil$ Steiner points along $\ell$, and let $v'$ (respectively, $b_{i,q}'$) be the Steiner point closest to $\tilde v$ (respectively, $\tilde b_{i,q}$).

\subparagraph{Construction.}
To meet the requirements for \Cref{thm:O1-SLT}, we need a isosceles triangle with points on the base, and an apex angle of $\Theta(\sqrt{\epsilon})$. Let $a$ and $b$ be the points where $\ell$ terminates and then let $\triangle=(s,a,b)$ be the triangle with $s$ as its apex. Therefore, the base of $\triangle$ is $\ell$, which has the necessary points. Earlier, we showed that the angle of the apex of the cone $f_i(S_i)$ is less than or equal to $\frac{\sqrt{\epsilon}}{1-c\sqrt{\epsilon}}$. We can see that for small epsilon this is on the order of $\Theta(\sqrt{\epsilon})$, since $\frac{\sqrt{\epsilon}}{1-c\sqrt{\epsilon}}$ is bounded below by $\sqrt{\epsilon}$ and above by $c'\sqrt{\epsilon}$, where $c'$ is sufficiently large. Therefore, we have met the angle requirement for the core example.

To finalize the construction, we will denote by $T^C_i$ the core example tree rooted at $s$. Let $E_i$ be the graph consisting of each edge $(b'_{i,q},b_{i,q})$. Finally, let $\tilde T_i$ be the union of $T^C_i$, $E_i$, and $H^*_i$. For the entire point set in $\R^d$, let $\hat{G}$ be the graph obtained from the union of all $k$ Steiner trees $f^{-1}_i(\tilde T_i),\ldots ,f^{-1}_k(\tilde T_k)$, where $f^{-1}_i: f_i(S_i)\to S_i$ is the inverse of $f_i$. Finally, let $T$ be an SPT of $\hat G$ rooted at $s$.

\subsection{Analysis}

\subparagraph{Root Stretch Analysis.}
We want to show that the root stretch of $T$ is at most $1+\epsilon$, which we will do by analyzing the component trees $\tilde T_i$. Specifically, we are going to show that the root stretch of each $\tilde T_i$ is $1+O(\epsilon)$, which can be improved to $1+\epsilon$ at the expense of increasing lightness by a constant factor. We will analyze this in three steps, beginning with the simplest case of the path $(s,v',v)$, then the path defined by the edges $(s,b'_{i,q})$, $(b'_{i,q},b_{i,q})$, and $d_{H^*}(b_{i,q},v_j)$, and finally the addition of the core example.

We make a few necessary observations. Note that the distance between any two Steiner points on $\ell$ is $w(\ell)\sqrt{\epsilon}$, and that $w(\ell)=c' d(s,r)$ for some constant $c'$. Therefore, the distance between Steiner points on $\ell$ is $w(\ell)\sqrt{\epsilon} = c'd(s,r)\sqrt{\epsilon} = O(\epsilon)d(s,r)$. This means that $d(\tilde v,v')=O(\epsilon)d(s,r)$ and notice that, by definition, $d(s,r)\leq d(s,v)$ for all $v\in V_i$. Finally, by our construction, for each point $v\in f_i(V_i)$, there is a break point $b_{i,q}\in B_i$ at most $\epsilon d_d(s,b_{i+1})$ away along the path $H_i^*$, which implies that $d_{H^*}(b_{i,q},v_j)\leq \epsilon d(s,b_{i+1})=O(\epsilon)d(s,r)$. We have enough to continue with the path $(s,v',v)$.
\begin{align*}
    w(s,v',v) &= d(s,v')+d(v',v) \\
    &\leq d(s,\tilde v) + d(\tilde v,v') +d(v',\tilde v)+d(\tilde v,v)\\
    &= d(s,v) + 2d(\tilde v,v') & (d(s,\tilde v)+d(\tilde v,v)=d(s,v)) \\
    &= d(s,v) + O(\epsilon)d(s,r) & (d(\tilde v,v')=O(\epsilon)d(s,r)) \\
    &\leq d(s,v) + O(\epsilon)d(s,v) & (d(s,r)\leq d(s,v)) \\
    &= (1+O(\epsilon))d(s,v) .
\end{align*}

This shows that the stretch of path $(s,v',v)$ is $1+O(\epsilon)$. Next we will relate the weights $w(s,b',b)$ and $w(s,v',v)$ for our final analysis
\begin{align*}
    w(s,b',b) &= d(s,b') + d(b',b) \\
    &\leq d(s,v') + d(v',b') + d(b',v') + d(v',v) + d(b,v) \\
    &= w(s,v',v)+2d(b',v')+d(b,v) \\
    &= w(s,v',v)+O(\epsilon)d(s,r) .
\end{align*}

Next, we will use these last two statements to show that $w(s,b',b)\circ d_{H^*}(b,v)$ is also a $1+O(\epsilon)$.
\begin{align*}
    w(s,b',b)\circ d_{H^*}(b,v) &= w(s,b',b) + d_{H^*}(b,v) \\
    &\leq w(s,v',v) + O(\epsilon)d(s,r) \\
    &\leq (1+O(\epsilon))d(s,v)  .
\end{align*}

Finally, in the construction, the SLT follows the core example tree $T^C_i$ instead of a direct line to $b'$. But, it will only grow by a factor of $1+\epsilon$, meaning that the path $d_{T_i^C}(s,b')\circ d(b',b)\circ d_{H^*}(b,v)$ will also be $1+O(\epsilon)$. Finally, we can increase the lightness by a constant factor to reduce the root stretch to $1+\epsilon$.

Now that we have shown it, we need to show that it holds in $\R^d$ as well. We have that for all $v\in f_i(V_i)$,  $d_{T_i}(s,v)\leq(1+\epsilon)d_2(f_i(s),f_i(v))$. Note that the shortest path distance in the surface $S_i$ between $s$ and any $v$ in the point set is a straight line in $\R^d$, since the edge between them is contained in $C_j$. Therefore, the shortest path distance in $S_i$ between $s$ and $v$ is equal to the Euclidean distance in $\R^d$; or equivalently, $d_{f_i(S_i)}(s,v)=d_d(s,v)$.
\begin{align*}
    d_{T_i}(s,v) &\leq (1+\epsilon)d_2(f_i(s),f_i(v)) \\
    &= (1+\epsilon)d_{S_i}(s,v) &\text{Isometry} \\
    &= (1+\epsilon)d_d(s,v) .
\end{align*}
This shows that the root stretch is preserved when each $\tilde T_i$ is isometrically mapped to the plane, meaning the root stretch of the final tree remains $1+\epsilon$.

\subparagraph{Lightness Analysis.}
Still working in $\R^d$, let $S$ be the set of edges $(s,b_i)$ between $s$ and points. Then, let the graph $G=(P,S\cup H^*)$ be obtained by adding all edges in $S$ to the path $H^*$. Finally, let $Q$ be a SPT over $G$, rooted at $s$. This is similar to the ``Phase~1'' SLT from~\cite{Solomon15}. Furthermore, note that $G$ can be split into $k$ trees consisting of the edge $(s,b_i)$ and the path $H^*_i$ without the final edge, which we will denote by $\tau_1,\ldots,\tau_k$. Each of these $\tau_i$ are spanning trees of each point set $V_i\cup\{s\}$.

It is enough to analyze the weight of $Q$, since the weight of the final SLT will be inherited from it. Also, since $Q$ is a SPT over $G$, its weight is bounded above by $w(G)$, so we only need to show the bound for $w(G)$.
By construction, we have $w(G)=w(H^*)+\sum_{i=1}^{k}d(s,b_i)$ and each $d(s,b_i)=\frac{1}{\sqrt{\epsilon}}d_H(b_i,b_{i+1})$. It follows that
\begin{align*}
    w(G) &= w(H^*)+\sum_{i=1}^{k}d(s,b_i) \\
    &= w(H^*)+\frac{1}{\sqrt{\epsilon}}\sum_{i=1}^{k}d_{H^*}(b_i,b_{i+1}) \\
    &\leq w(H^*)+\frac{1}{\sqrt{\epsilon}}\sum_{j=1}^{n}d_{H^*}(v_j,v_{j+1}) \\
    &= w(H^*)+\frac{1}{\sqrt{\epsilon}}w(H^*) \\
    &=\left(1+\frac{1}{\sqrt{\epsilon}}\right)w(H^*) .
\end{align*}

Recall that $w(H^*)\leq2\cdot w(M)=2\cdot w(\MST(P))$. Combining these inequalities, we obtain
\[
    w(G)\leq\left(1+\frac{1}{\sqrt{\epsilon}}\right)w(H^*)\leq\left(1+\frac{1}{\sqrt{\epsilon}}\right)2\cdot w(M)=\left(1+\frac{1}{\sqrt{\epsilon}}\right)2\cdot w(\MST(P)),
\]
which implies that the lightness of $G$ is $O(\epsilon^{-1/2})$. Lastly, the weight of our final tree $\hat G$ is inherited from $G$:
\[
    w(\hat G)
    =\sum w(f^{-1}_i(\tilde T_i))
    =\sum O(w(\tau_i))
    =O\left(\sum w(\tau_i)\right)
    =O(w(G)).
\]

\subparagraph{Conclusion.}
From the root stretch and lightness analysis, we have shown that the final SLT $T$ has root stretch $1+\epsilon$, and that the graph it is a SPT of $\hat G$, has $O(\epsilon^{-1/2})$ lightness. Because $T$ is a SPT of $\hat G$, we know that it has the same lightness, meaning that $T$ is an $(1+\epsilon,O(\sqrt{1/\epsilon}))$-SLT for $V$ rooted at $s$.

\section{2-Dimensional Core Example}
\label{sec:Core}

In this section, we revisit the core example proposed by Solomon~\cite[Section~2.1]{Solomon15}, provide a detailed lightness and root stretch analysis, and make some necessary adjustments to the parameters originally proposed in~\cite{Solomon14,Solomon15}.

    The core example consists of a point set on the boundary of an isosceles triangle $\triangle abs$ with apex $s$ and $\angle asb=\sqrt{\epsilon}$. Specifically, the point set contains $s$ and a finite set of points on the line segment $ab$; see  \Cref{fig:triangle}. Without loss of generality, we can scale the triangle such that $d(s,a)=d(s,b)=1$; hence $d(a,b)=2 \sin(\alpha/2)$ and the height of the triangle is $\cos(\alpha/2)$.
     %finitely many collinear points on s are $s$ and at the origin and a given set $V=\{s =v_0, v_1=a,...,v_n=b\}$, illustrated in \Cref{fig:triangle}, where the apex vertex of $\triangle sab$ is the root $s \in V$, its apex angle is defined as $\alpha := \sqrt{\epsilon}$, and the points $V' = V \setminus \{s\}$ are located along the triangular base. For the triangle's sides, let $d(s,a)=d(s,b)=1$. The base length of $\triangle$ is given by $d(a,b)=2 \sin(\alpha/2)$ and the height is given by $d(s,o)=\cos(\alpha/2)$.

\begin{figure}[h]
        \centering
        \includegraphics[height=7cm]{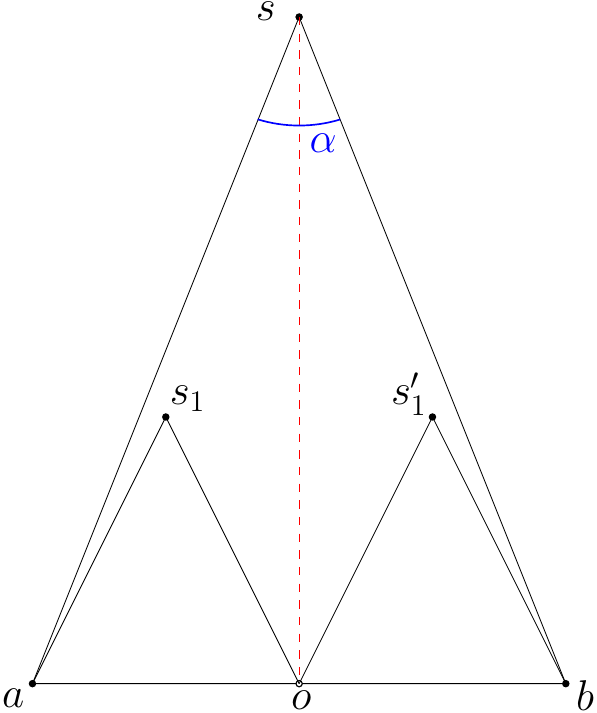}
        \hspace{8mm}
         \includegraphics[height=7cm]{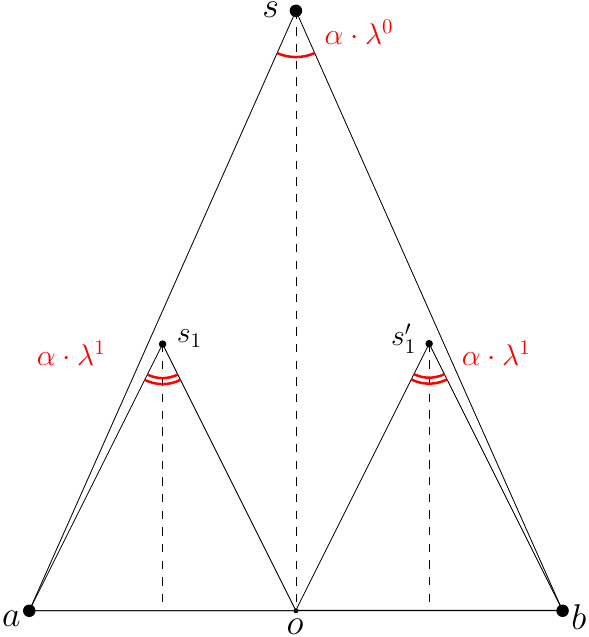}
        \caption{Left: The isosceles triangle $\triangle sab$ with the apex $s$ and apex angle $\alpha$. The base of $\triangle$ contains the points in $V'$ and the Steiner points $\mathcal{P}$.
        Right: $\triangle sab$ with the adjusted apex angles for the first step in the recursion.}
        \label{fig:triangle}
    \end{figure}

    The main result of this section is that for a desirable point set, there exists a Steiner $(1+\epsilon, O(1))$-SLT. We construct an SLT by recursively subdividing the base $ab$, and creating smaller isosceles triangles inside $\triangle asb$. In each iteration, the apex angles increase by a factor of $\lambda>1$. We summarize the result as follows.
   \begin{theorem}\label{thm:O1-SLT}
   Let $\epsilon<\pi/4$, and let $\Delta sab$ be an isosceles triangle with $d(a,s)=d(b,s)$ and $\angle asb=\sqrt{\epsilon}$. Then for a set $V$ of $n+1$ points, $n$ points on the line segment $ab$ and a source $s\in V$, there exists a Steiner $(1+\epsilon, O(1))$-SLT.
    \end{theorem}
   \subparagraph{Construction.}
    We may assume without loss of generality that the side $ab$ of $\Delta sab$ lies on the $x$-axis such that its midpoint is the origin $o$, and $s$ is on the positive $y$-axis.
    Let $\alpha=\sqrt{\epsilon}$; see \Cref{fig:triangle}. We may further assume, by a suitable scaling, that $d(a,s)=d(b,s)=1$. This implies that $d(o,s)=\cos \frac{\alpha}{2}$ and $d(a,b)=2\sin\frac{\alpha}{2}$. Let $V'=V\setminus \{s\}$, and assume that
    $V'=\{v_1,\ldots , v_n\}$ labeled by increasing $x$-coordinates.

    Let $\lambda\geq 1$ be a parameter to be specified later, and let $k=\lceil \log_2 \sqrt{1/\epsilon}\rceil +1 $.
    We recursively define sets of congruent isosceles triangles $\mathcal{T}_i$ and point set $\mathcal{S}_i$ for $i=0,\ldots , k$, as follows.
    For $i=0$, let $\mathcal{T}_0=\{\Delta asb\}$ and $\mathcal{S}_0=\{s\}$.
    For $i=1,\ldots ,k-1$, proceed as follows.
    Initialize $\mathcal{T}_{i+1}=\emptyset$ and $\mathcal{S}_{i+1}=\emptyset$.
    For each triangle $\Delta a's'b'\in \mathcal{T}_i$, partition the line segment $a'b'$ into two segments of equal length, $a'o'$ and $o'b'$. Let $\Delta a's_{i+1}c'$ and $\Delta c's_{i+1}' b'$ be isosceles triangles with apex angle $\angle a's_{i+1}c' =\angle c's_{i+1}' b'=\alpha\cdot \lambda^i$. We add both new triangles to $\mathcal{T}_{i+1}$ by setting $\mathcal{T}_{i+1}:=\mathcal{T}_{i+1}\cup \{\Delta a's_{i+1}c', \Delta c's_{i+1}' b'\}$, and add their apex vertices to $\mathcal{S}_{i+1}$ by $\mathcal{S}_{i+1}:=\mathcal{S}_{i+1}\cup\{s_{i+1},s_{i+1}'\}$. Finally, let $\mathcal{B}$ be the set of all vertices of the triangles in $\mathcal{T}_k$ in along the line $ab$.

    We can now construct a geometric graph $G^{\text{core}}$ as follows. The vertex set of $G^{\text{core}}$ is $V\cup \mathcal{B}\cup\bigcup_{i=0}^k\mathcal{S}_i$. Consider a simple path connecting all vertices of $G^{\text{core}}$ that lie on the segment $ab$, and add this path to $G^{\rm core}$. For every $i\in \{0,\ldots , k-1\}$, connect the apex $s'$ of every triangle $\Delta a's' b'\in \mathcal{T}_i$ to the apices of  $\Delta a's_ic', \Delta c's_i' b'\in \mathcal{T}_{i+1}$;
    and connect the apex $s'$ of every triangle $\Delta a's' b'\in \mathcal{T}_k$
    to the base vertices $a'$ and $b'$. Finally, let $T^{\text{core}}$ be a shortest path tree (SPT) of $G^{\mathrm{core}}$ rooted at $s$. This completes the construction of $T^{\text{core}}$.

    \subparagraph{Angle Calculations.}
    We first observe that at the end of recursion, the apex angle of the smallest triangles in $\mathcal{T}_1$ is still less than $\pi/2$, and the base of each triangle in $\mathcal{T}_k$ is less than $\epsilon/2$. Assuming $\epsilon < 1$, the final angle of the recursion will be at most
    \begin{equation}\label{eq:angle}
     \sqrt{\epsilon}\cdot \lambda^k
     %\leq \sqrt{\epsilon}\cdot\lambda^{\lceil \log_2 \sqrt{1/\epsilon}\rceil}
      \leq \sqrt{\epsilon}\cdot (\sqrt{1/\epsilon})^{\log_\lambda 2} \cdot\lambda^2
        =\epsilon^{\frac12(1-\log_\lambda 2)} \cdot \lambda^2
        <  \lambda^2
        <\frac{\pi}{2}
    \end{equation}
    if $\lambda <\sqrt{\pi/2}<1.26$. The base of the isosceles triangles at the end of the recursion is bounded by
 \begin{equation*}
      \frac{d(a,b)}{2^k}
        %\frac{\sin\left(\frac{\sqrt{\epsilon}}{2}\right)}{2^{i-1}} =
        \frac{d(a,b)}{2^{\lceil \log_2 \sqrt{1/\epsilon}\rceil+1}}
        <\frac{2\sin\left(\frac{\sqrt{\epsilon}}{2}\right)}{2\cdot 2^{\log_2\sqrt{1/\epsilon}}} %\\
       % = \frac{2\sin\left(\frac{\sqrt{\epsilon}}{2}\right)}{\sqrt{1/\epsilon}} %\\
        < \frac{\frac{\sqrt{\epsilon}}{2}}{\sqrt{1/\epsilon}} %\\
        %= \frac{\sqrt{\epsilon}}{2\sqrt{1/\epsilon}} %\\
        =\frac{\epsilon}{2} .
    \end{equation*}

    \subparagraph{Lightness Analysis.} The lightness of $T^{\text{core}}$ is the ratio $\frac{w(T^{\text{core}})}{w(\MST(V))}$. For an upper bound on lightness, we need an upper bound for the weight of $T^{\text{core}}$ and a lower bound for that of $\mathrm{MST}(V)$. The lower bound for the MST is easy, since $d(s,a)=1$ by construction, and so $d(s,v)\geq \cos\frac{\alpha}{2}$ for every $v\in V'$.
    This implies $w(\MST(V))\geq \cos\frac{\alpha}{2}$.

    The remainder of this section is dedicated to bounding the weight of $T^{\text{core}}$, and specifically finding that it has $O(1)$ lightness. Consider a level of the recursion $\T_i\neq\emptyset$. On each level $i$, $T^\text{core}$ has $2^{i+1}$ edges, each of length $d(s_i,s_{i+1})$. Note that $d(s_i,s_{i+1})\leq d(s_i,b)$, because each $s_{i+1}$ is contained in a triangle with apex $s_i$ and diameter $d(s_i,b)$. We analyze the distance $d(s_i,b)$ in each level of the recursion.
    There are $2^i$ congruent isosceles triangles in the set $\T_i$, and the lengths of their bases are $d(a,b)/2^i=d(o,b)/2^{i-1}=\sin(\frac{\alpha}{2})/2^{i-1}$.
    We can also express the length of each base as $2\sin(\alpha\lambda^i/2)\cdot d(a,s_i)=2\sin(\alpha\lambda^i/2)\cdot d(b,s_i)$. Equating these expressions yields
    \begin{align*}
     2\sin(\alpha\lambda^i/2)\cdot d(b,s_i) &= \sin\left(\frac{\alpha}{2}\right)/2^{i-1} \\
     d(b,s_i) &=\frac{\sin\left(\alpha/2\right)}{2^i\sin(\alpha\lambda^i/2)} .
    \end{align*}
     Taylor estimates $\frac{x}{2}\leq \sin(x)\leq x$ for all $0\leq x\leq \frac{\pi}{2}$ (using \Cref{eq:angle}) give an upper bound
 \begin{equation}
 d(b,s_i)
 =\frac{\sin\left(\alpha/2\right)}{2^i\sin(\alpha\lambda^i/2)}
 \leq \frac{\alpha/2}{2^i\cdot \alpha \lambda^i/4}
 = \frac{1}{2^{i-1}\cdot \lambda^i}.
 \end{equation}
    The binary tree $T^{\text{core}}$ has $2^i$ edges of length $d(b,s_i)$ at level $i$, so their combined length is bounded by
    \begin{equation}
    \label{eq:base-ratio}
        2^{i+1} d(s_{i+1}s_i)
        \leq 2^{i+1}d(b,s_i)
        \leq 2^{i+1}\cdot \frac{1}{2^{i-1}\cdot \lambda^i}
        = \frac{4}{\lambda^i} .
    \end{equation}
    Summation over all levels gives a bound for the total weight of $T^{\text{core}}$:
    \begin{equation}
        w(T^{\text{core}})
        = \sum_{i=0}^{k} 2^{i+1}d(s_{i+1},s_i)
        \leq \sum_{i=0}^{k} 2^{i+1}d(b,s_i)
        \leq \sum_{i=0}^{k}\frac{4}{\lambda^i}
        < \sum_{i=0}^{\infty}\frac{4}{\lambda^i}
        = \frac{4}{1-\frac{1}{\lambda}}
        =\frac{4\lambda}{\lambda-1}.
    \end{equation}
    The weight of $T^{\text{core}}$ is bounded above by a geometric series. This series converges for any $\lambda>1$. Since $w(\mathrm{MST}(V))\leq \cos\frac{\alpha}{2}$, then the lightness of $T^{\text{core}}$
    is bounded by $w(T^{\text{core}})  =\frac{4\lambda}{(\lambda-1)\cos(\alpha/2)} =O(1)$.

    \subparagraph{Root-Stretch Analysis.}
    The root stretch for a vertex $v\in V'$ is the ratio of the distance $d_T(s,v)$ along the tree to the Euclidean distance $d(s,v)$. Every root-to-leaf path in $T^{\rm core}$ includes exactly one edge from each level. Since all edges on level $i$ have the same length, $d(s_i,s_{i+1})$, then $d_T(s,v)=\sum_{i=0}^{k} d(s_i,s_{i+1})$.
    Note that $d(s,v)\geq |\mathrm{proj}_y(sv)| = \cos\frac{\alpha}{2}$, where $\mathrm{proj}_y(uv)$ denoted the orthogonal projection of segment $uv$ to the $y$-axis.
    For every $v\in V'$, the $sv$-path in $T^{\text{core}}$ is $y$-monotone, and so the length of its $y$-projection is exactly $|\mathrm{proj}_y(sv)| = \cos\frac{\alpha}{2}$. For a line segment $uv$, we define
    \[
         \text{slack}(uv)=d(u,v)-|\mathrm{proj}_y(uv)|.
    \]
    Then the length of an $sv$-path in $T^{\text{core}}$ is
    \[
    d_T(s,v)
    =\sum_{i=0}^{k} d(s_is_{i+1})
    =\sum_{i=0}^{k} \big(|\mathrm{proj}_y(s_is_{i+1})|+\mathrm{slack}(s_is_{i+1})\big)
    = \cos\frac{\alpha}{2} +\sum_{i=0}^{k}\mathrm{slack}(s_is_{i+1}) .
    \]

   Furthermore, notice that $d(a,o_i)=\frac{1}{2^{i+1}}d(a,b)=\frac{1}{2^{i+1}}2\sin\frac{\alpha}{2}$, and similarly to the lightness $d(s_i,s_{i+1})\leq d(s_i,a)$. We can bound the slack from above by
    \begin{align*}
        \text{slack}(s_is_{i+1}) &\leq d(s_i, a) - (s_i,o_i) \\
        &= \frac{d(a,o_i)}{\sin\left(\frac{\alpha}{2} \cdot \lambda^i\right)} - \frac{d(a, o_i)}{\tan\left(\frac{\alpha}{2} \cdot \lambda^i \right)} \\
        %&= d(a,o_i) \cdot \frac{1}{\sin\left(\frac{\alpha}{2} \cdot \lambda^i\right)} - d(a,o_i) \cdot \frac{\cos\left(\frac{\alpha}{2} \cdot \lambda^i \right)}{\sin\left(\frac{\alpha}{2} \cdot \lambda^i \right)} \\
        &= d(a,o_i)\cdot \frac{\left( 1 - \cos\left(\frac{\alpha}{2} \cdot \lambda^i \right) \right)}{\sin\left(\frac{\alpha}{2}  \cdot \lambda^i\right)}\\
        &= \frac{\sin\left(\frac{\alpha}{2}\right)}{2^i} \cdot \frac{\left( 1 - \cos\left(\frac{\alpha}{2} \cdot \lambda^i \right) \right)}{\sin\left(\frac{\alpha}{2} \cdot \lambda^i\right)} \\
         &\leq \frac{\alpha/2}{2^i} \cdot \frac{(\alpha\lambda^i/2)^2/2}{(\alpha\lambda^i/2)/2} = \frac{\alpha^2}{4}\cdot \left(\frac{\lambda}{2}\right)^i,
    \end{align*}
where the last line used Taylor estimates $\frac{x}{2}\leq \sin(x)\leq x$ and  $\cos(x)\geq 1-\frac{x^2}{2}$.
%    \begin{equation*}
%        \frac{\sin\left(\frac{\alpha}{2}\right)}{2^i} \cdot \frac{\left( 1 - \cos\left(\frac{\alpha}{2} \cdot \lambda^i \right) \right)}{\sin\left(\frac{\alpha}{2} \cdot \lambda^i\right)}
%       \leq \frac{\alpha/2}{2^i} \cdot \frac{(\alpha\lambda^i/2)^2/2}{(\alpha\lambda^i/2)/2}
%= \frac{\alpha^2}{4}\cdot \left(\frac{\lambda}{2}\right)^i
%  \end{equation*}
%   Altogether, we have that the slack for each iteration of the recursion is bounded above by $\frac{\alpha^2}{4}\cdot \left(\frac{\lambda}{2}\right)^i$:
Summation over all levels yields
    \begin{equation}
    \label{eq:root-stretch1}
        d_T(s,v)
    = \cos\frac{\alpha}{2} +\sum_{i=0}^{k}\mathrm{slack}(s_is_{i+1})
    \leq \cos\frac{\alpha}{2} +\frac{\alpha^2}{4}\cdot\sum_{i=0}^{\infty} \left(\frac{\lambda}{2}\right)^i
    \leq \cos\frac{\alpha}{2} + \frac{\alpha^2}{2(2-\lambda)}
    \end{equation}
for all $1<\lambda<2$. We have all the information needed to estimate the root stretch, so using the Taylor estimate $d(s,v)\geq \cos(\alpha/2)\geq 1-\alpha^2/8$ we find for any point in $\mathcal{B}$ we get a stretch of:
    \[
    \frac{d_T(s,v)}{d(s,v)}
    \leq \frac{\cos(\alpha/2) +\frac{\alpha^2}{2(2-\lambda)}}{\cos (\alpha/2)}
    \leq 1+\frac{\alpha^2}{2(2-\lambda)(1-\alpha^2/8)}.
    \]

    We can now substitute $\alpha^2=\epsilon$ and fix $\lambda=5/4$, and additionally observe that the largest distance between any point in $\mathcal{B}$ an arbitrary point on $ab$ is $\epsilon/4$. We can now show that we obtain the desired root stretch for every point $v\in V'$:
    \begin{equation}
    \label{eq:root-stretch2}
    \frac{d_T(s,v)}{d(s,v)}
    \leq  \frac{d_T(s,v)}{d(s,v)}
    \leq \frac{\cos(\alpha/2) +\frac{\epsilon}{2(2-\lambda)}+\frac{\epsilon}{4}}{\cos (\alpha/2)}
    \leq 1+\frac{\frac{2}{3}+\frac{1}{4}}{1-\epsilon/8} \cdot \epsilon
     <1+\epsilon.
    \end{equation}
This completes the proof of \Cref{thm:O1-SLT}.

\subparagraph{Remark.} Even though the lightness analysis works for any $1<\lambda<2$, the root stretch analysis would not work for parameter $\lambda=\frac32$, which was proposed by Solomon~\cite{Solomon14,Solomon15}. Specifically, \Cref{eq:root-stretch2} shows that $\lambda=\frac32$ would not guarantee root-stretch $1+\epsilon$ even for the base points in $\mathcal{B}$. The range $1<\lambda<\frac54$ was also needed for the angle constraints in \Cref{eq:angle}, which was necessary for our Taylor estimates. It is possible, though, that the construction with $\lambda=\frac32$ still produces a Steiner $(1+\epsilon,O(1))$-SLTs, however, it would require a more careful analysis.

\section{Core Example in $d$-Dimensions}
\label{sec:pyramid}

    In this section, we extend the above construction to $d$-dimensions, using right pyramids of cube bases. We will find that there is no change in root stretch when moving to the $d$-dimensional version, but the lightness depends exponentially on $d$. A \emph{right pyramid} is the convex hull of a $(d-1$-dimensional hypercube $Q$ and an apex $s$ such that the orthogonal projection of $s$ to the hyperplane spanned by $Q$ is the midpoint of $Q$. Observe that the triangle $\Delta asc$, formed by $s$ and a body diagonal $ac$ of $Q$, is isosceles. We define the \emph{apex angle} of the pyramid as $\angle asc$. Note that for a given $(d-1)$-dimensional cube $Q$, the apex angle $\angle asc$ determines the location of the apex $s$ (up to reflection in the hyperplane spanned by $Q$). In the construction below, we use $(d-1)$-dimensional hypercubes $Q$ in a hyperplane orthogonal to the $x_1$-axis and apex angles to define pyramids.

    \begin{theorem}\label{thm:O2-SLT}
   Let $Q$ be the base cube and $s$ be the apex of a right pyramid of apex angle $\sqrt{\epsilon}$ with base $Q$ in $\mathbb{R}^d$. Let $V$ be a set $V$ of $n+1$ points, $n$ of which are in a $n^{1/(d-1)}\times\ldots \times n^{1/(d-1)}$ grid in $Q$ and a source point $s$. For all
   $n\geq (2\sqrt{d}\cdot \epsilon^{0.66-d/2})^{(d-1)/(d-2)}$,
   %$n\geq \epsilon^{-1.68}$,
   there exists a Steiner $(1+\epsilon,O(1))$-SLT.
    \end{theorem}

    \subparagraph{Construction.}
    Let $\alpha=\sqrt{\epsilon}$. We may assume that the right pyramid $P$ with base $Q$ is centered at the origin $o$ and, by a suitable scaling, that $s$ is at unit distance from all corners of $Q$. This implies that $d(o,s)=\cos \frac{\alpha}{2}$ and $d(a,c)=2\sin\frac{\alpha}{2}$, where $ac$ is a body diagonal of $Q$.
    Let $V'=V\setminus \{s\}$, and  $V'=\{v_1,\ldots , v_n\}$ labeled by lexicographic ordering (according to the $x_1,\ldots x_{d-1}$-coordinates).

    Let $\lambda\geq 1$ be a parameter to be specified later and let $k=\lceil \log_2 \sqrt{1/\epsilon}\rceil +1 $.
    We recursively define sets of congruent pyramids $\mathcal{P}_i$ and point set $\mathcal{S}_i$ for $i=0,\ldots , k-1$, as follows.
    For $i=0$, let $\mathcal{P}_0=\{P\}$ and $\mathcal{S}_0=\{s\}$.
   % While $\alpha\cdot \lambda^{i+1}<\frac{\pi}{2}$, proceed as follows:
    For $i=0,\ldots ,\lceil \log_2 \sqrt{1/\epsilon}\rceil$, proceed as follows.
    Initialize $\mathcal{P}_{i+1}=\emptyset$ and $\mathcal{S}_{i+1}=\emptyset$.
    For each pyramid $P\in \mathcal{P}_i$, subdivide the base cube into $2^{d-1}$  congruent cubes, and let their respective pyramids have apex angle $\alpha\cdot \lambda^i$. We add all new pyramids to $\mathcal{P}_{i+1}$, and add their apex vertices to $\mathcal{S}_{i+1}$. Let $\mathcal{B}$ be the set of vertices of the pyramids in $\mathcal{P}_k$ the lie in $Q$.

    We can now construct a geometric graph $G^{d\text{-core}}$ as follows. The vertex set of $G^{d\text{-core}}$ is $V\cup \mathcal{B}\cup\bigcup_{i=0}^k\mathcal{S}_i$.
    Consider a $\frac54$-spanner for the point set $V\cup \mathcal{B}$ in the hyperplane spanned by $Q$, and add it to $G^{d\text{-core}}$. For every $i\in \{0,\ldots , k-1\}$, connect the apex $s'$ of every pyramid $P\in \mathcal{P}_i$ to the apexes of the four pyramids in $\mathcal{P}_{i+1}$ contained in $P$;
    and connect the apex $s'$ of every pyramid $P\in \mathcal{P}_k$
    to the $2^{d-1}$ base vertices of $P$.  Finally, let $T^{d\text{-core}}$ be a shortest path tree (SPT) of $G^{d\text{-core}}$ rooted at $s$. This completes the construction of $T^{d\text{-core}}$.

    \subparagraph{Root-Stretch Analysis.}
    As a reminder, the root stretch for a vertex $v\in V'$ is the ratio of the distance $d_T(s,v)$ along the tree to the Euclidean distance $d(s,v)$. Every root-to-leaf path in $T^{d\text{-core}}$ includes exactly one edge from each level. Since all edges on level $i$ have the same length, $d(s_i,s_{i+1})$, then $d_T(s,v)=\sum_{i=0}^{k} d(s_i,s_{i+1})$.

    Consider the diagonal cross section of $T^{\text{d-core}}$ between $s$ and two opposite corners of $Q$, without loss of generality say $a,c$. $\triangle asc$ has apex angle $\alpha$, each level of the recursion is $2^i$ triangles with angle $\alpha\cdot\lambda^i$, and finally $d(a,c)=2\sin\frac{\alpha}{2}$. It should be clear that this cross section is exactly $T^\text{core}$ from \cref{sec:Core}. As we said above, all edges on the level $i$ have the same length, which means that any root-to-leaf path in $T^{d\text{-core}}$ has the same length as a root-to-leaf path in the diagonal cross section, which is equivalent to $T^{\rm core}$. Consequently, \Cref{eq:root-stretch1} still holds for the path in $T^{d\text{-core}}$ from $s$ to any  point in $\mathcal{B}$.

   Consider an arbitrary point $v\in V'$: it is contained in a pyramid $P\in \mathcal{P}_k$, so there is a point $u\in \mathcal{B}$ at distance at most $\frac12\cdot d(a,b)/2^k \leq \sin(\frac{\sqrt{\epsilon}}{2})/(2\cdot \sqrt{1/\epsilon})
  <\frac{\sqrt{\epsilon}}{2}/(2\sqrt{1/\epsilon})<\epsilon/4$. The $\frac54$-spanner on $V\cup\mathcal{B}$ contains an $uv$-path of length at most $\frac54(\epsilon/4)=\frac{5}{16}\cdot \epsilon$. So, for every point in $v\in V'$, the graph $G^{d\text{-core}}$ contains an $sv$-path of length at most $\cos(\alpha/2)+\frac{3\epsilon}{4}+5\epsilon/16 =\cos(\alpha/2)+\frac{17}{16}\cdot \epsilon$. Now \Cref{eq:root-stretch2} holds, as well, establishing that the root stretch of $T^{d\text{-core}}$ is at most $1+\epsilon$ for all $v\in V'$.

    \subparagraph{Lightness Analysis.}
    Remember that the lightness of $T^{d\text{-core}}$ is the ratio $\frac{w(T^{d\text{-core}})}{w(\MST(V))}$.
    Recall that $V'\subset V$ is a set of $n$ points in a $n^{1/(d-1)}\times\ldots \times n^{1/(d-1)}$ grid in a $(d-1)$-dimensional hyercube $Q$ of diagonal $d(a,c)=2\sin\frac{\sqrt{\epsilon}}{2}$ and side length $d(a,b)=\frac{2}{\sqrt{d}}\sin\frac{\sqrt{\epsilon}}{2}>\frac{2}{\sqrt{d}}\cdot \frac{\sqrt{\epsilon}}{4}= \frac{\sqrt{\epsilon/d}}{2}$ if $\sqrt{\epsilon}<\pi$. The distance between any two adjacent grid points is $d(a,b)/n^{1/(d-1)}> \sqrt{\epsilon/d}/(2n^{1/(d-1)})$; this is a lower bound on any edge of $\MST(V)$. Consequently, $w(\MST(V))> n\cdot \sqrt{\epsilon/d}/(2n^{1/(d-1)}) = \sqrt{\epsilon/4d}\cdot n^{(d-2)/(d-1)}> \Omega(\epsilon^{1.16-d/2})$ for the range of $n$ specified in the theorem.

    In the lightness analysis in \Cref{sec:Core}, there were $2^i$ congruent triangles in each set $\mathcal{T}_i$, but for this construction we have $2^{(d-1)i}$ congruent pyramids in each $\mathcal{P}_i$. As argued above in the root stretch analysis, we know that the length of $d(b,s_i)$ in $T^{d\text{-core}}$ is the same as in $T^{\rm core}$ for all $i$, therefore much of the earlier lightness analysis holds, except that we have $2^{(d-1)}$ edges in each level of the recursion. So we will start from \Cref{eq:base-ratio}, at each recursion we have
    \begin{equation}
    \label{eq:base-ratio-3D}
        2^{(d-1)(i+1)} d(s_{i+1}s_i)
        \leq 2^{(d-1)(i+1)}d(b,s_i)
        \leq 2^{(d-1)(i+1)}\cdot \frac{1}{2^{i-1}\cdot \lambda^i}
        = 2^d\cdot \frac{2^{(d-2)i}}{\lambda^i} .
    \end{equation}
    Summation over all levels gives
    \begin{align*}
    \sum_{i=0}^{k} 2^{(d-1)(i+1)}d(s_{i+1},s_i)
        %& \leq \sum_{i=0}^{k} 4^{i+1}d(b,s_i) \\
        & \leq 2^d\sum_{i=0}^{k}\frac{2^{(d-2)i}}{\lambda^i} \\
        &\leq O\left(2^d \cdot \left(\frac{2^{d-2}}{\lambda}\right)^k\right) \\
        &= O\left(2^d \cdot\left(\frac{2^{d-2}}{\lambda}\right)^{\log_2 \sqrt{1/\epsilon}}\right)\\
        %&= O\left(\left(\frac{2}{\lambda}\right)^{\log_{2/\lambda} \sqrt{1/\epsilon}/\log_{2/\lambda}2}\right)\\
        &= O\left(2^d \cdot \left(\frac{1}{\epsilon}\right)^{1/(2\log_{2^{d-2}/\lambda}2)}\right)\\
        &= O\left(2^d \cdot \left(\frac{1}{\epsilon}\right)^{\frac12 \log_2(2^{d-2}/\lambda))}\right)\\
        &= O\left(2^d \cdot \left(\frac{1}{\epsilon}\right)^{\frac12 (d-2-\log_2\lambda)}\right)\\
        &< O\left(2^d \cdot \left(\frac{1}{\epsilon}\right)^{d/2-1.16}\right)
    \end{align*}
for $\lambda=\frac54$. The weight of a $\frac{5}{4}$-spanner for $V'$ in the plane is known to be $O(\MST(V))$~\cite{DasHN93}; see also \cite{20M1317906}.
Overall, the weight of $G^{\text{d-core}}$ is $O(\MST(V)+2^d/\epsilon^{d/2-1.16})$.
Together with the earlier calculated weight of the $\MST(V)$, we have that the lightness of $T^{\text{d-core}}$ is bounded by
\begin{equation*}
    \frac{w(T^{\text{d-core}})}{w(\MST(V))}
    \leq \frac{O(\MST(V)+2^d/\epsilon^{d/2-1.16})}{w(\MST(V))}
    \leq 1+ 2^d\cdot \frac{O\left(\epsilon^{1.16-d/2}\right)}{\Omega(\epsilon^{1.16-d/2})} =O(1) .
\end{equation*}

%\bibliographystyle{plainurl}
%\bibliography{ref}

\end{document}